\RequirePackage[2020-02-02]{latexrelease}
\documentclass[onecolumn,amsmath,amssymb,10pt,aps]{revtex4}
  
\usepackage[utf8]{inputenc}
\usepackage[T1]{fontenc}

\usepackage{amsmath}
\usepackage{amsfonts}
\usepackage{graphicx}
\usepackage{yfonts}
\usepackage{color}
\usepackage[normalem]{ulem}
\usepackage{amsthm}
\usepackage{bm}
\usepackage{bbm}
\usepackage{mathtools}
\usepackage{array}
\usepackage{placeins}
\usepackage{enumitem}
\usepackage{tikz}

\def\<{\langle}
\def\>{\rangle}
\newcommand{\tr}{\mathrm{Tr}}
\newcommand{\Tr}{\mathrm{Tr}}
\def\oper{{\mathchoice{\rm 1\mskip-4mu l}{\rm 1\mskip-4mu l}
{\rm 1\mskip-4.5mu l}{\rm 1\mskip-5mu l}}}
\DeclareMathAlphabet\mathbfcal{OMS}{cmsy}{b}{n}

\mathchardef\mhyphen="2D 

\newtheorem{Theorem}{Theorem}

\newtheorem{Remark}{Remark}

\begin{document}

\title{Enhancing phase-covariant channel performance with non-unitality}
	
	\author{Katarzyna Siudzi{\'n}ska$^{1}$ and Micha{\l} Studzi\'nski$^{2}$}
	\affiliation{
	    $^1$ Institute of Physics, Faculty of Physics, Astronomy and Informatics \\  Nicolaus Copernicus University in Toru\'{n}, ul. Grudzi\k{a}dzka 5, 87--100 Toru\'{n}, Poland\\
	    $^2$ Institute of Theoretical Physics and Astrophysics, University of Gda\'nsk, National Quantum Information Centre, 80-952 Gda\'nsk, Poland}

\begin{abstract}
We analyze quantum communication properties of phase-covariant channels depending on their degree of non-unitality. In particular, we derive analytical formulas for minimal and maximal channel fidelity on pure states and maximal output purity. Next, we introduce a measure of non-unitality and show how to manipulate between unital and maximally non-unital maps by considering classical mixtures of quantum channels. Finally, we prove that maximal fidelity and maximal output purity increase with non-unitality and present several examples. Interestingly, non-unitality can also prolong quantum entanglement and lead to its rebirth.
\end{abstract}

\flushbottom

\maketitle

\thispagestyle{empty}

\section{Introduction}

Quantum channels, which describe transformations between input and output states, are present in every quantum information and communication processing tasks. However, physical channels are inherently noisy, making their possible applications limited. Therefore, extensive research on optimal methods of information transmission is essential for quantum communication. One way to enhance the amount of reliably transmitted information is to reduce the effects of noise. Among the methods specially tailored for this task are error correction, error mitigation, and error suppression techniques~\cite{Review,QEC}. Another way to approach the problem of detrimental noise is to instead use the noise as a resource \cite{Verstraete,zanardi17,Engineering_capacity,fidelity}. This way, we can enhance the quantities that measure channel transmission properties, like fidelity, purity, or capacity.

A full characterization of quantum channel properties is in general a very challenging undertaking. To make the problem more tractable, one can introduce additional symmetries, like the covariance property of channels. By definition, a quantum channel $\Lambda$ is covariant with respect to unitary representations $U$, $V$ of a finite (or compact) group $G$ if
\begin{equation}
\label{eq:covgen}
\Lambda\left[U(g)\rho U^{\dagger}(g)\right]=V(g)\Lambda[\rho]V^{\dagger}(g)\qquad \forall \ g\in G,
\end{equation}
for any valid density operator $\rho$. Such maps are also known as {\it $G$-covariant}. The seminal results come from Scutaru~\cite{Scutaru}, who proved a Stinespring-type theorem in the $C^{\ast}$-algebraic framework, giving a base for more applicative research. In particular, SU(2)-covariant channels were used to describe entanglement in spin systems~\cite{Schliemann} and dimerization of quantum spin chains~\cite{Nachtergaele2017}. In quantum information, covariant channels help to analyze additivity property of the Holevo capacity~\cite{651037} and minimal output entropy~\cite{irr-cov1,irr-cov2,Fan1,irr-cov3,7790801}. The covariance property also allows to prove strong converse properties for the classical capacity~\cite{PhysRevLett.103.070504} and entanglement-assisted classical capacity~\cite{Datta2016}. There are also known methods on how to construct positive covariant maps~\cite{Kopszak,Studz}.

A special class of $U(1)$-covariant maps consists in phase-covariant qubit maps, for which $U(\phi)=V(\phi)=\exp(-i\sigma_3\phi)$, where $\sigma_3=\rm{diag}(1,-1)$, $\phi\in\mathbb{R}$. Phase-covariant channels provide evolutions that combine pure dephasing with energy absorption and emission~\cite{phase-cov-PRL,phase-cov}. At first, they were introduced phenomenologically in the description of thermalization and dephasing processes that go beyond the Markovian approximation \cite{PC1}. The associated dynamical equations were later derived microscopically for a weakly-coupled spin-boson model under the secular approximation \cite{PC3}. Phase-covariant maps were applied in the contexts of quantum speed evolution~\cite{QSTCov}, non-Markovianity of quantum evolution~\cite{e23030331}, quantum optics~\cite{Marvian_2013,RevModPhys.79.555}, and quantum metrology~\cite{PhysRevA.94.042101}. They play a substantial role in the description of phase covariant devices~\cite{Buscemi:07} and quantum cloning machines~\cite{PhysRevA.62.012302}.

The main goal of our paper is to prove that transition performance can be improved by allowing for non-unitality of quantum channels. This is shown on the example of fidelity and purity measures, which determine the distortion between input and output states. We start with analytical derivations of formulas for minimal and maximal channel fidelity on pure states, as well as maximal output purities in terms of Schatten $p$-norms. The pure states that correspond to the respective extremal values are also provided. Next, we ask about the evolution of quantum entanglement under the assumption that half of a maximally entangled state is sent through the phase-covariant channel. For an entanglement measure, we choose concurrence, which is also related to entanglement formation.

In the main part, we provide important applications for our results. We consider families of quantum channels that differ only by the non-unitality degree. By comparing the fidelity and purity measures, we show that unital maps always display the worst performance for every analyzed measure except the minimal channel fidelity. Actually, this drop in channel performance is monotonically decreasing with the degree of non-unitality -- that is, the closer the channel is to being unital, the smaller the increase of the corresponding measure. Similar behavior is observed for concurrence and entanglement of formation, which measure entanglement between two qubits. In the presented examples, we observe not only how to prolong entanglement but also how to speed up its rebirth after sudden death. Finally, we also show how to engineer the desired non-unitality degree with a classical mixture of the unital and maximally non-unital phase-covariant quantum map. In this case, the probability distribution can be treated as the noise beneficial for the properties of quantum evolution. Finally, it is important to note that the enhanced performance of non-unital channels is observed at any moment in time. This is a novelty compared to previous works on noise suppression by counteracting its effects with another form of noise~\cite{Klesse,fidelity,Engineering_capacity}, where the positive effects were only temporary.

\section{Phase-covariant channels}
\label{sec:PhaseCovChann}

Consider a class of qubit maps covariant with respect to phase rotations on the Bloch sphere, which are represented by a unitary transformation
\begin{equation}
\label{eq:conds}
U(\phi)=\exp(-i\sigma_3\phi),\qquad\phi\in\mathbb{R},\qquad\sigma_3=\begin{pmatrix}
1 & 0\\0 & -1\end{pmatrix}.
\end{equation}
Such maps are called {\it phase-covariant} and satisfy the covariance condition
\begin{equation}
\label{eq:covohase}
\Lambda\left[U(\phi)\rho U^{\dagger}(\phi)\right]
=U(\phi)\Lambda[\rho]U^{\dagger}(\phi)\qquad \forall \ \phi\in \mathbb{R}
\end{equation}
for any input density operator $\rho$. Note that $U(\phi)$ defines a continuous group parameterized with an angle $\phi$. Up to the unitary transformation $\rho\mapsto\exp(-i\sigma_3\theta)\rho\exp(i\sigma_3\theta)$, $\theta\in\mathbb{R}$, the most general form of $\Lambda$ reads~\cite{phase-cov,phase-cov-PRL}
\begin{equation}
\label{eq:actionofLambda}
\Lambda[\rho]=\frac 12 \left[(\mathbb{I}+\lambda_{\ast}\sigma_3)\Tr\rho
+\lambda_1\sigma_1\Tr(\rho\sigma_1)+\lambda_1\sigma_2\Tr(\rho\sigma_2)
+\lambda_3\sigma_3\Tr(\rho\sigma_3)\right],
\end{equation}
where $\sigma_1$, $\sigma_2$, $\sigma_3$ denote the Pauli matrices.
The real numbers $\lambda_1$ and $\lambda_3$ are the eigenvalues of $\Lambda$ to the following eigenvectors,
\begin{equation}
\Lambda[\sigma_1]=\lambda_1\sigma_1,\qquad
\Lambda[\sigma_2]=\lambda_1\sigma_2,\qquad
\Lambda[\sigma_3]=\lambda_3\sigma_3.
\end{equation}
The last parameter, $\lambda_{\ast}$, is responsible for non-unitality -- that is, failing to preserve the identity operator $\mathbb{I}$ ($\Lambda(\mathbb{I})\neq \mathbb{I}$). It also determines the map's invariant state ($\Lambda[\rho_{\ast}]=\rho_{\ast}$), which reads
\begin{equation}\label{rhoast}
\rho_{\ast}=\frac{1}{2}\left[\mathbb{I}+\frac{\lambda_{\ast}}{1-\lambda_3}\sigma_3\right].
\end{equation}
For $\lambda_\ast=0$, one recovers a symmetric subclass of Pauli channels, which are unital qubit maps. In this case, the invariant state $\rho_\ast=\mathbb{I}/2$ is maximally mixed. Therefore, one can say that the non-unitality property of phase-covariant channels is controlled by $\lambda_{\ast}$.
Finally, to ensure that $\Lambda$ is a quantum channel (completely positive, trace-preserving map), its parameters have to satisfy the conditions \cite{phase-cov}
\begin{equation}
|\lambda_\ast|+|\lambda_3|\leq 1,\qquad 4\lambda_1^2+\lambda_\ast^2\leq (1+\lambda_3)^2.
\end{equation}

\section{Performance measures of quantum channels}
\label{sec:PerfMeas}

\subsection{Channel fidelity}

In quantum information theory, the fidelity measures the distance that separates two quantum states \cite{Nielsen,Zyczkowski}. Therefore, it can be used to determine their distinguishability. According to Uhlmann's definition \cite{Uhlmann}, the fidelity between states represented by the density operators $\rho$ and $\sigma$ is given by
\begin{equation}\label{statesfidelity}
F(\rho,\sigma):=\left(\Tr\sqrt{\sqrt{\rho}\sigma\sqrt{\rho}}\right)^2.
\end{equation}
Observe that $0\leq F(\rho,\sigma)\leq 1$, and the equality $F(\rho,\sigma)=1$ holds if and only if $\rho=\sigma$. This formula served as a starting point to introduce the channel fidelity $F(\rho,\Lambda[\rho])$; that is, the fidelity between its input $\rho$ and output $\Lambda[\rho]$ states \cite{Raginsky}. It measures the distortion of an initial state under the use of a channel $\Lambda$. For pure inputs represented by rank-1 projectors $P$, the channel fidelity is bounded by the minimal and maximal fidelity on pure input states \cite{Sommers3},
\begin{equation}\label{channelfidelity}
\begin{split}
f_{\min}(\Lambda)&=\min_PF(P,\Lambda[P])=\min_P \Tr(P\Lambda[P]),\\
f_{\max}(\Lambda)&=\max_PF(P,\Lambda[P])=\max_P \Tr(P\Lambda[P]).
\end{split}
\end{equation}
Due to its concavity property, the minimal fidelity for mixed inputs is reached on a pure state. Hence, $f_{\min}(\Lambda)$ is also the minimal channel fidelity on mixed states. However, $f_{\max}(\Lambda)$ is not the maximal fidelity in general, as the maximal value $\max_\rho F(\rho,\Lambda[\rho])=F(\rho_\ast,\Lambda[\rho_\ast])=1$ is reached on the invariant state $\rho_\ast$ of $\Lambda$.

\begin{Theorem}\label{Th1}
The minimal and maximal channel fidelities on pure input states under the action of a phase-covariant channel are given by the following formulas,
\begin{equation}\label{fidmin}
f_{\min}(\Lambda)=\left\{
\begin{aligned}
&\frac 12 \left(1+\lambda_1-\frac{\lambda_{\ast}^2}{4(\lambda_3-\lambda_1)}\right)\qquad{\rm for}\qquad \lambda_3>\lambda_1,\, |\lambda_\ast|\leq 2(\lambda_3-\lambda_1),\\
&\frac 12 (1+\lambda_3-|\lambda_{\ast}|)\qquad{\rm otherwise},
\end{aligned}\right.
\end{equation}
\begin{equation}\label{fidmax}
f_{\max}(\Lambda)=\left\{
\begin{aligned}
&\frac 12 \left(1+\lambda_1+\frac{\lambda_{\ast}^2}{4(\lambda_1-\lambda_3)}\right)\qquad{\rm for}\qquad \lambda_3<\lambda_1,\, |\lambda_\ast|\leq 2(\lambda_1-\lambda_3),\\
&\frac 12 (1+\lambda_3+|\lambda_{\ast}|)\qquad{\rm otherwise}.
\end{aligned}
\right.
\end{equation}
\end{Theorem}

\begin{proof}
Take a pure state represented by a rank-1 projector
\begin{equation}\label{proj}
P=\frac 12 \left(\mathbb{I}+\sum_{k=1}^3x_k\sigma_k\right),
\end{equation}
where $x_k$ are real numbers such that $\sum_{k=1}^3x_k^2=1$.
The action of a phase-covariant channel $\Lambda$ onto $P$ produces
\begin{equation}\label{akcja}
\Lambda[P]=\frac 12 \left(\mathbb{I}+\lambda_{\ast}\sigma_3 +\lambda_1x_1\sigma_1+\lambda_1x_2\sigma_2+\lambda_3x_3\sigma_3\right),
\end{equation}
and hence the corresponding channel fidelity is, by definition,
\begin{equation}
F(P,\Lambda[P])=\Tr(P\Lambda[P])=\frac 12 [1+\lambda_1(x_1^2+x_2^2)+\lambda_3x_3^2+\lambda_{\ast}x_3].
\end{equation}
The next step is to find $x_3$ that minimize or maximize $F$. For simplicity, we introduce the function
\begin{equation}
G(P,\Lambda[P])=2F(P,\Lambda[P])-1=\lambda_1(1-x_3^2)+\lambda_3x_3^2+\lambda_{\ast}x_3
=(\lambda_3-\lambda_1)x_3^2+\lambda_{\ast}x_3+\lambda_1,
\end{equation}
which reaches its extremal values for the same parameters $x_3$ as the function $F(P,\Lambda[P])$. From now on, we consider $G(P,\Lambda[P])$ as the function of $x_3$ with fixed channel parameters and denote it by $G(x_3)$. Recall that extremal points of a function are found by equating the first derivative to zero and checking the sign of the second derivative at that point. In our case, the first derivative results in a single critical point,
\begin{equation}\label{x3}
G^\prime(x_3)=2(\lambda_3-\lambda_1)x_3+\lambda_{\ast}=0\qquad\implies\qquad
x_3=-\frac{\lambda_{\ast}}{2(\lambda_3-\lambda_1)}.
\end{equation}
Note that $|x_3|\leq 1$, which, together with eq. (\ref{x3}), gives us an additional constraint for the channel parameters,
\begin{equation}
|\lambda_\ast|\leq 2|\lambda_3-\lambda_1|.
\end{equation}
Calculating the second derivative yields
\begin{equation}
G^{\prime\prime}\left(x_3=-\frac{\lambda_{\ast}}{2(\lambda_3-\lambda_1)}\right)
=2(\lambda_3-\lambda_1).
\end{equation}
This way, we found a local minimum for $\lambda_3>\lambda_1$ and a local maximum for $\lambda_3<\lambda_1$. For $\lambda_3=\lambda_1$, $G(x_3)$ is a linear function, so there are no local extrema. Due to the domain of $x_3$ being closed, the global extrema of $G$ are reached either at the local extremal points or at the endpoints $x_3=\pm 1$, where $G(x_3=\pm 1)=\lambda_3\pm\lambda_{\ast}$. Our results can be summarized as follows,
\begin{equation}
f_{\min}(\Lambda)=\left\{
\begin{aligned}
&\frac 12 \min\left\{1+\lambda_1-\frac{\lambda_{\ast}^2}{4(\lambda_3-\lambda_1)},
1+\lambda_3-|\lambda_{\ast}|\right\}\qquad{\rm for}\qquad \lambda_3>\lambda_1,\, |\lambda_\ast|\leq 2(\lambda_3-\lambda_1),\\
&\frac 12 (1+\lambda_3-|\lambda_{\ast}|)\qquad{\rm otherwise},
\end{aligned}\right.
\end{equation}
\begin{equation}
f_{\max}(\Lambda)=\left\{
\begin{aligned}
&\frac 12 \max\left\{1+\lambda_1+\frac{\lambda_{\ast}^2}{4(\lambda_1-\lambda_3)},
1+\lambda_3+|\lambda_{\ast}|\right\}\qquad{\rm for}\qquad \lambda_3<\lambda_1,\, |\lambda_\ast|\leq 2(\lambda_1-\lambda_3),\\
&\frac 12 (1+\lambda_3+|\lambda_{\ast}|)\qquad{\rm otherwise}.
\end{aligned}
\right.
\end{equation}
Finally, observing that the first term in the curly brackets is always minimal for $f_{\min}(\Lambda)$ and maximal for $f_{\max}(\Lambda)$, we recover the formulas from eqs. (\ref{fidmin}--\ref{fidmax}).
\end{proof}

\begin{Remark}
The minimal and maximal channel fidelities on pure inputs are reached on a one-parameter family of rank-1 projectors. Every projector $P$ is characterized by three parameters $x_k$, $k=1,2,3$, from which only $x_3$ is fixed in the process of finding the extremal points of $F(P,\Lambda[P])$. Due to the constraint $\sum_{k=1}^3x_k^2=1$ from eq. (\ref{proj}), $x_2$ depends on the choice of $x_1$. Hence, we are left with a free parameter $x_1$ that changes between $\pm\sqrt{1-x_3^2}$.
\end{Remark}

For $\lambda_\ast=0$, one recovers the formulas for the Pauli channels \cite{norms,fidelity},
\begin{equation}
f_{\min}(\Lambda)=\frac{1+\lambda_{\min}}{2},\qquad
f_{\max}(\Lambda)=\frac{1+\lambda_{\max}}{2},
\end{equation}
where $\lambda_{\max}=\max_{k=1,3}\lambda_k$ and $\lambda_{\min}=\min_{k=1,3}\lambda_k$. Note that, contrary to the case with $\lambda_\ast\neq 0$, $f_{\min}(\Lambda)$ and $f_{\max}(\Lambda)$ depend only on a single eigenvalue.

\subsection{Maximal output purity}

The purity measures how close a given state is to a pure state. This question can also be applied to quantum channels $\Lambda$, where one checks the purity of the output $\Lambda[P]$ for pure inputs $P$. The higher the purity of the output, the less distorted is the input. However, one is usually interested in the best case scenario, which corresponds to the maximal output purity. This property is measured by the maximal output $p$-norm defined by
\begin{equation}
\nu_p(\Lambda):=\max_P\|\Lambda[P]\|_p,
\end{equation}
where the Schatten $p$-norm reads \cite{TQI,Bhatia}
\begin{align}
&\|\Lambda[P]\|_p:=(\Tr\Lambda[P]^p)^{1/p},\qquad 1\leq p<\infty,\\
&\|\Lambda[P]\|_\infty:=\max_Q\Tr(Q\Lambda[P]),
\end{align}
and $P$ and $Q$ are a rank-1 projectors. Here, let us consider two of the most popular choices: $p=2$ and $p=\infty$.

\begin{Theorem}\label{Th2}
The maximal output $2$-norm of phase-covariant channels satisfies
\begin{equation}\label{nu2}
\nu_2^2(\Lambda)=\left\{
\begin{aligned}
&\frac 12 \left(1+\lambda_1^2
+\frac{\lambda_1^2\lambda_\ast^2}{\lambda_1^2-\lambda_3^2}\right),\qquad |\lambda_1|>|\lambda_3|,\,|\lambda_3\lambda_\ast|\leq \lambda_1^2-\lambda_3^2,\\
&\frac 12 (1+\lambda_3^2+\lambda_\ast^2+2|\lambda_3\lambda_\ast|)\qquad{\rm otherwise}.
\end{aligned}
\right.
\end{equation}
\end{Theorem}

\begin{proof}
Using eq. (\ref{akcja}), we find
\begin{equation}
\tr(\Lambda[P]^2)=\frac 12 \left[1+(\lambda_3^2-\lambda_1^2)x_3^2+2\lambda_3\lambda_\ast x_3+\lambda_1^2+\lambda_\ast^2\right].
\end{equation}
Define an auxiliary function $K(\Lambda):=2\tr(\Lambda[P]^2)-1$ whose extremal points coincide with that of $\nu_2^2(\Lambda)$. To find the extremas of $(\Lambda)\equiv K(x_3)$, we calculate the first and second derivatives with respect to $x_3$;
\begin{align}
K^\prime(x_3)&=2(\lambda_3^2-\lambda_1^2)x_3+2\lambda_\ast \lambda_3=0\qquad\implies\qquad  x_3=-\frac{\lambda_3\lambda_\ast}{\lambda_3^2-\lambda_1^2},\label{eq:funK1}\\
K^{\prime\prime}(x_3)&=2(\lambda_3^2-\lambda_1^2).
\end{align}
Since $|x_3|\leq 1$, the above formula for $x_3$ provides an additional constraint on the channel parameters,
\begin{equation}
|\lambda_3\lambda_\ast|\leq |\lambda_3^2-\lambda_1^2|.
\end{equation}
Now, if $\lambda_3^2>\lambda_1^2$, then we obtain the local minimum, whereas $\lambda_3^2<\lambda_1^2$ gives rise to the local maximum. However, if $\lambda_3^2=\lambda_1^2$, then there are no local extrema. In this case, the global extremal points are reached on the endpoints of the domain $x_3=\pm 1$, where the function $K$ takes the values
\begin{equation}
K(x_3=\pm 1)=\lambda_\ast^2+\lambda_3^2\pm 2\lambda_3\lambda_\ast.
\end{equation}
Therefore, the formula for the maximal output 2-norm reads
\begin{equation}
\nu_2^2(\Lambda)=\left\{
\begin{aligned}
&\frac 12 \max\left\{1+K\left(x_3=-\frac{\lambda_3\lambda_\ast}{\lambda_3^2-\lambda_1^2}\right),
1+\lambda_\ast^2+\lambda_3^2+2|\lambda_3\lambda_\ast|\right\},\qquad |\lambda_1|>|\lambda_3|,\,|\lambda_3\lambda_\ast|\leq |\lambda_1^2-\lambda_3^2|,\\
&\frac 12 \max\{1+\lambda_\ast^2+\lambda_3^2+2|\lambda_3\lambda_\ast|\}\qquad{\rm otherwise}.
\end{aligned}
\right.
\end{equation}
Observing that, in the range provided by the first line of this equation,
\begin{equation}
K\left(x_3=-\frac{\lambda_3\lambda_\ast}{\lambda_3^2-\lambda_1^2}\right)=\lambda_1^2
\left(1+\frac{\lambda_\ast^2}{\lambda_1^2-\lambda_3^2}\right)\geq
\lambda_\ast^2+\lambda_3^2+2|\lambda_3\lambda_\ast|,
\end{equation}
we finally arrive at eq. (\ref{nu2}).
\end{proof}

After putting $\lambda_\ast=0$, one recovers the squared maximal output $2$-norm for the Pauli channels \cite{Ruskai},
\begin{equation}
\nu_2^2(\Lambda)=\frac 12 \left[1+\max_\alpha\lambda_\alpha^2\right].
\end{equation}
Unlike in the formula for $\lambda_\ast\neq 0$, here $\nu_2^2(\Lambda)$ depends only on the squared channel parameters.

\begin{Theorem}\label{Th3}
The maximal output $\infty$-norm of phase-covariant channels is equal to
\begin{equation}\label{nuinf}
\nu_{\infty}(\Lambda)=\frac 12 \left[1+\max\{|\lambda_1|,|\lambda_3\pm\lambda_\ast|\}\right].
\end{equation}
\end{Theorem}

\begin{proof}
Let us take two rank-1 projectors,
\begin{equation}
P=\frac 12 \left(\mathbb{I}+\sum_{k=1}^3x_k\sigma_k\right),\qquad
Q=\frac 12 \left(\mathbb{I}+\sum_{k=1}^3y_k\sigma_k\right),\qquad
\sum_{k=1}^3x_k^2=\sum_{k=1}^3ky_k^2=1.
\end{equation}
In what follows, we make use of the trace condition $0\leq\Tr(PQ)\leq 1$, which is equivalent to
\begin{equation}
-1\leq\sum_{k=1}^3x_ky_k\leq 1.
\end{equation}
On the other hand, we find
\begin{equation}\label{trqp}
\Tr(Q\Lambda[P])=\frac 12 \left[1+\lambda_\ast y_3+\lambda_1x_1y_1
+\lambda_1x_2y_2+\lambda_3x_3y_3\right].
\end{equation}
From the form of $\Tr(Q\Lambda[P])$, it is easy to see that it has no local extrema in the projectors' parameters (due to being a linear function in all $x_k$, $y_k$). Hence, the global maximum is reached on one of the edges: $x_k=\pm 1$, $y_k=\pm 1$. After making this substitution in eq. (\ref{trqp}), one arrives at eq. (\ref{nuinf}).
\end{proof}

The formula for the maximal output $\infty$-norm
\begin{equation}
\nu_{\infty}(\Lambda)=\frac 12 \left[1+\max_{\alpha=1,3}|\lambda_\alpha|\right]
\end{equation}
for $\lambda_\ast=0$ was derived in \cite{norms}. There, it was also observed that for the Pauli channels one has $\nu_\infty=f_{\max}$ if $\max\lambda_\alpha=\max|\lambda_\alpha|$. Interestingly, an analogical comparison exists for phase-covariant channels. Namely, $\nu_\infty=f_{\max}$ if either $\lambda_\ast=0$ or $\lambda_3\geq|\lambda_1|+|\lambda_\ast|$.

\subsection{Concurrence}

Assume that we extend our qubit system by composing it with another qubit system. The first subsystem evolves according to a phase-covariant channel while the second subsystem remains unchanged. If initially the qubit pair is maximally entangled, then the total state changes according to $\rho_W\mapsto\rho_W^\prime=(\oper\otimes\Lambda)[\rho_W]$, where $\rho_W=(1/4)\sum_{i,j=0}^1 |ii\>\<jj|$. The entanglement of formation between two qubit systems can be measured using Wootters' concurrence \cite{Wooters1,Wooters2}
\begin{equation}
c(\rho)=\max\{0,\sqrt{r_1}-\sqrt{r_2}-\sqrt{r_3}-\sqrt{r_4}\},
\end{equation}
where $r_1\geq r_2\geq r_3\geq r_4$ are the eigenvalues of $X(\rho):=\rho(\sigma_2\otimes\sigma_2)\overline{\rho}(\sigma_2\otimes\sigma_2)$. 
Observe that, in terms of the Pauli matrices, the state $\rho_W$ has the form
\begin{equation}
\rho_W=\frac{1}{4}\left(\mathbb{I}\otimes\mathbb{I}+\sigma_1\otimes \sigma_1-\sigma_2\otimes \sigma_2+\sigma_3\otimes \sigma_3\right),
\end{equation}
and therefore it is straightforward to show that its evolution is given by
\begin{equation}
\rho_W^\prime=\frac{1}{4}\left(\mathbb{I}\otimes\mathbb{I}
+\lambda_{\ast}\mathbb{I}\otimes\sigma_3
+\lambda_1\sigma_1\otimes\sigma_1-\lambda_1\sigma_2\otimes\sigma_2
+\lambda_3\sigma_3\otimes\sigma_3\right).
\end{equation}
In the computational basis, $X[\rho_W^\prime]$ is represented by the matrix
\begin{equation}
X[\rho_W^\prime]=
\frac{1}{16}\begin{pmatrix}
4\lambda_1^2+(1+\lambda_3)^2-\lambda_\ast^2 & 0 & 0 & 4\lambda_1(1+\lambda_3+\lambda_\ast) \\
0 & (1-\lambda_3)^2-\lambda_\ast^2 & 0 & 0 \\
0 & 0 & (1-\lambda_3)^2-\lambda_\ast^2 & 0 \\
4\lambda_1(1+\lambda_3-\lambda_\ast) & 0 & 0 & 4\lambda_1^2+(1+\lambda_3)^2-\lambda_\ast^2
\end{pmatrix},
\end{equation}
whose eigenvalues read
\begin{equation}
\begin{split}
R_1=R_2=\frac{1}{16}\Big[(1-\lambda_3)^2-\lambda_\ast^2\Big],\qquad
R_\pm=\frac{1}{16}\Big[2\lambda_1\pm
\sqrt{(1+\lambda_3)^2-\lambda_\ast^2}\Big]^2.
\end{split}
\end{equation}
Due to $R_+\geq R_1=R_2\geq R_-$, the corresponding formula for concurrence reduces to
\begin{equation}\label{crhoW}
c[\rho_W^\prime]=\frac{1}{2}\max\left\{0,2|\lambda_1|-\sqrt{(\lambda_3-1)^2-\lambda_\ast^2}\right\}.
\end{equation}
If one takes $\lambda_\ast=0$, the above equation reproduces the concurrence
\begin{equation}
c[\rho_W^\prime]=\max\{0,2|\lambda_1|+\lambda_3-1\}
\end{equation}
of the Pauli channels satisfying $\Lambda[\sigma_2]=\lambda_1\sigma_2$.

\section{Applications: Using non-unitality to improve channel performance}
\label{sec:Applications}

The measures of fidelity, purity, and entanglement derived in previous section depend on the channel eigenvalues $\lambda_1$, $\lambda_3$, as well as the parameter $\lambda_\ast$ that vanishes for unital channels. Therefore, a question arises: given two quantum maps, one unital and one non-unital, can we determine which one has a better performance in quantum communication tasks according to those measures?
To answer this, consider two phase-covariant qubit channels: a unital (Pauli) channel $\Lambda_{\rm U}$ and a non-unital channel $\Lambda_{\rm NU}$. Assume that these channels have common eigenvalues and share three eigenvectors, so that
\begin{equation}
\Lambda_{\rm U}[\sigma_k]=\lambda_k\sigma_k,\qquad 
\Lambda_{\rm NU}[\sigma_k]=\lambda_k\sigma_k,\qquad k=1,2,3\qquad (\lambda_2\equiv\lambda_1).
\end{equation}
The final eigenvector (to the eigenvalue $\lambda_0=1$) is associated with the invariant state of the channel and depends on the value of $\lambda_\ast$. For $\Lambda_{\rm U}$, the invariant state is the maximally mixed state $\rho_0=\mathbb{I}/2$. However, for $\Lambda_{\rm NU}$, the invariant state is instead given by $\rho_\ast$ in eq. (\ref{rhoast}).

Due to the fact that $\Lambda_{\rm U}$ and $\Lambda_{\rm NU}$ differ only in one parameter, we can easily compare the results of Section 3 for the corresponding measures.

\begin{Remark}
Non-unital phase-covariant channels present a better performance than their unital counterparts when the maximal fidelity, maximal output purity, and concurrence are measured. The opposite behavior is observed for the minimal fidelity, which decreases for non-zero $\lambda_\ast$.
\end{Remark}

From now on, assume that $\Lambda_{\rm NU}$ is {\it maximally non-unital}; that is, its parameter $\lambda_\ast$ admits the highest absolute value
\begin{equation}
|\lambda_\ast|=1-|\lambda_3|,
\end{equation}
which follows from the complete positivity conditions for the phase-covariant channels. Now, to construct a channel $\Lambda$ with intermediate values of $\lambda_\ast$, we take convex combinations of $\Lambda_{\rm NU}$ and $\Lambda_{\rm U}$. The resulting channel
\begin{equation}\label{mix}
\Lambda=(1-p)\Lambda_{\rm U}+p\Lambda_{\rm NU},\qquad 0\leq p\leq 1,
\end{equation}
shares its eigenvalues with both $\Lambda_{\rm U}$ and $\Lambda_{\rm NU}$. Moreover, the parameter that characterizes its non-unitality satisfies the formula
\begin{equation}
\lambda_\ast^{\pm}=\pm p(1-|\lambda_3|).
\end{equation}
Hence, $\lambda_\ast^{\pm}$ can be treated as a measure of $\Lambda(t)$'s non-unitality. The greater $p$ we take, the more non-unital is the mixture, with $p=0$ and $p=1$ corresponding to the unital and maximally non-unital maps, respectively. This notion can be generalized to all non-unital phase-covariant maps.

\begin{Remark}
For phase-covariant qubit maps, we introduce the measure of non-unitality
\begin{equation}
\mathrm{NU}(\Lambda)=\frac{|\lambda_\ast|}{1-|\lambda_3|},
\end{equation}
which determines their degree of non-unitality. In particular, if $\mathrm{NU}(\Lambda)=1$, then $\Lambda$ is maximally non-unital. On the other hand, $\mathrm{NU}(\Lambda)=0$ corresponds to unital maps.
\end{Remark}

In general, quantum channels are used to describe the dynamics of open quantum systems; that is, systems that interact with an external environment. Continuous time-evolution is provided by dynamical maps, which are time-parameterized quantum channels $\Lambda(t)$ with the initial condition $\Lambda(0)=\oper$. Such maps are often solutions of dynamical equations called the {\it master equations}. Quantum systems with memoryless (Markovian) evolution satisfy the semigroup master equation $\dot{\Lambda}(t)=\mathcal{L}\Lambda(t)$ with a constant generator $\mathcal{L}$. The presence of strong system-environment interactions makes it necessary to consider more complicated equations, e.g. with time-dependent generators or memory kernels. For our considerations, however, the explicit form of master equations does not matter.

In what follows, we consider examples of phase-covariant dynamical maps given by eq. (\ref{mix}) and analyze the evolution of their purity, fidelity, and concurrence of the evolved maximally entangled state $\rho_W$.

\subsection{Example 1 -- Exponential decay}

\FloatBarrier

Let us consider a maximally non-unital dynamical map $\Lambda_{\rm NU}(t)$ characterized by
\begin{equation}\label{eigexp}
\lambda_1(t)=e^{-t},\qquad \lambda_3(t)=e^{-2t},\qquad \lambda_\ast(t)=1-e^{-2t},
\end{equation}
and a unital map $\Lambda_{\rm U}(t)$ that shares eigenvalues with $\Lambda_{\rm NU}(t)$. Observe that both channels are Markovian semigroups, where $\Lambda_{\rm NU}(t)$ corresponds to amplitude damping and $\Lambda_{\rm U}(t)$ to anisotropic dephasing. For any $0<p<1$, the mixture $\Lambda(t)$ is not a semigroup itself \cite{CC_GAD}. It is straightforward to derive the formulas for
\begin{itemize}
\item the minimal and maximal fidelities:
\begin{equation}
f_{\min}[\Lambda(t)]=\frac 12 [1-p+(1+p)e^{-2t}],\qquad f_{\max}[\Lambda(t)]=\left\{
\begin{aligned}
&\frac{1-e^{-2t}}{4(1-e^{-t})}[2+p^2\sinh t]
\qquad{\rm for}\qquad p\leq\frac{1-e^{-t}}{\sinh t},\\
&\frac 12 [1+p+(1-p)e^{-2t}]\qquad{\rm for}\qquad p>\frac{1-e^{-t}}{\sinh t};
\end{aligned}
\right.
\end{equation}
\item the maximal output purities:
\begin{equation}
\nu_2^2[\Lambda(t)]=\frac{1+p^2}{2}+\frac{1-p^2}{2}e^{-2t},\qquad \nu_\infty[\Lambda(t)]=\left\{
\begin{aligned}
&\frac 12 (1+e^{-t})
\qquad{\rm for}\qquad p\leq\frac{1-e^{-t}}{2\sinh t},\\
&\frac 12 [1+p+(1-p)e^{-2t}]\qquad{\rm for}\qquad p>\frac{1-e^{-t}}{2\sinh t}.
\end{aligned}
\right.
\end{equation}
\end{itemize}
Note that $f_{\min}$ and $\nu_2^2$ are given by simple expressions with exponential decay. In contrast, the formulas for $f_{\max}$ and $\nu_\infty$ are much more involved, having two potential outcomes depending on the values of $t$ and $p$. Despite the range conditions being implicit functions of time, both $f_{\max}$ and $\nu_\infty$ are continuous, as
\begin{equation}
\lim_{t\to t_\ast^{\pm}}f_{\max}[\Lambda(t)]=e^{-t_\ast}(1+\sinh t_\ast),\qquad
\lim_{t\to t_\ast^{\pm}}\nu_\infty[\Lambda(t)]=\frac 12 (1+e^{-t_\ast}),\qquad
\frac{1-e^{-t_\ast}}{\sinh t_\ast}\equiv p.
\end{equation}

Our results are plotted in Fig.\ref{exp}. It is clear that all the measures asymptotically decay in time and the curves corresponding to distinct values of $p$ cross only at $t=0$. Moreover, the minimal fidelity monotonically decreases with the increase of $p$, whereas all the other measures monotonically increase with $p$. Moreover, $f_{\max}$ and $\nu_\infty$ have the same asymptotical values. For $p=1$, $f_{\max}$, $\nu_2$, and $\nu_\infty$ reach their maximal value of 1. The only function that ever drops to zero is $f_{\min}[\Lambda_{\rm U}(t\to\infty)]$.

\begin{figure}[htb!]
   \includegraphics[width=0.8\textwidth]{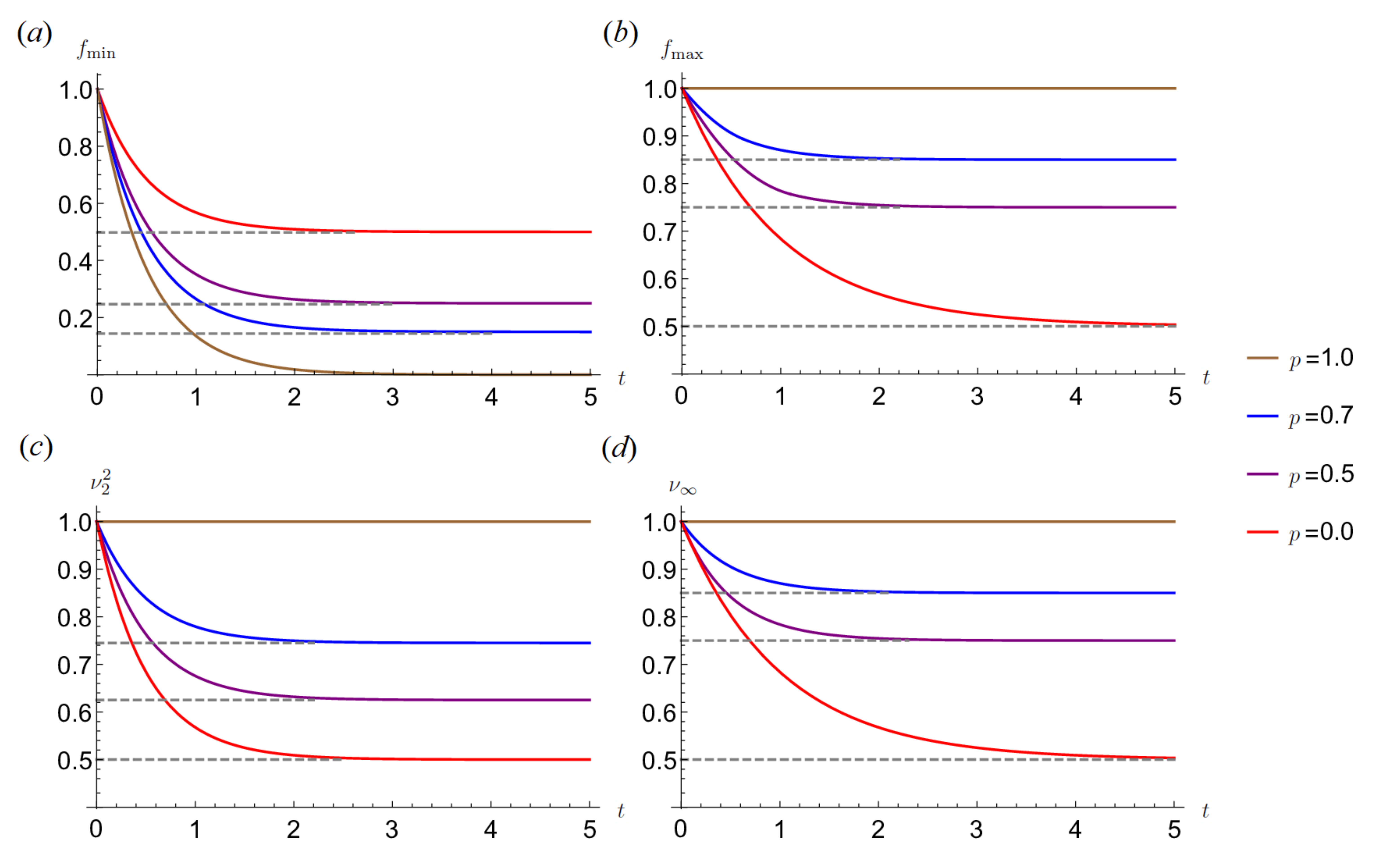}
\caption{Plots for exponentially decaying channel eigenvalues representing time-evolution of the minimal fidelity ($a$), maximal fidelity ($b$), maximal output 2-norm ($c$), and maximal output $\infty$-norm ($d$). The color curves correspond to the channel mixtures with $p=0$ (red), $p=0.5$ (purple), $p=0.7$ (blue), and $p=1$ (yellow).}
\label{exp}
\end{figure}

\FloatBarrier

\subsection{Example 2 -- Oscillations}

\FloatBarrier

This time, we take the dynamical map $\Lambda_{\rm NU}(t)$ whose parameters oscillate according to
\begin{equation}\label{eigosc}
\lambda_1(t)=\cos t,\qquad \lambda_3(t)=\cos^2t,\qquad \lambda_\ast(t)=\sin^2t,
\end{equation}
together with the unital map $\Lambda_{\rm U}(t)$ with the exact same eigenvalues. Note that both channels are non-invertible due to $\lambda_k(t)$ vanishing for finite times. The oscillatory behaviour manifests itself also in the associated measures, as one finds
\begin{itemize}
    \item the minimal fidelity:
  \begin{equation}
      f_{\min}[\Lambda(t)]=\left\{
\begin{aligned}
&\frac{\sin^2t(4\cos t+p^2\sin^2t)}{8\cos t(1-\cos t)}
\qquad{\rm for}\qquad -1<\cos t<0,\,p\leq\frac{2|\cos t|}{\sin^2t}(1-\cos t),\\
&\frac 12 [1+\cos^2t-p\sin^2t]\qquad{\rm otherwise};
\end{aligned}
\right.\\
  \end{equation}
    \item the maximal fidelity:
    \begin{equation}
       f_{\max}[\Lambda(t)]=\left\{
\begin{aligned}
&\frac{\sin^2t(4\cos t+p^2\sin^2t)}{8\cos t(1-\cos t)}
\qquad{\rm for}\qquad 0<\cos t<1,\,p\leq\frac{2\cos t}{\sin^2t}(1-\cos t),\\
&\frac 12 [1+\cos^2t+p\sin^2t]\qquad{\rm otherwise};
\end{aligned}
\right.\\ 
    \end{equation}
    \item the maximal output purities:
    \begin{equation}
        \nu_2^2[\Lambda(t)]=\frac 12 (1+\cos^2t+p^2\sin^2t),\qquad
\nu_\infty[\Lambda(t)]=\frac 12 (1+\max\{|\cos t|,|\cos^2t+p\sin^2t|\}).
    \end{equation}
\end{itemize}
This time, all the measures except for $\nu_2^2$ are given by relatively complicated formulas, even though the parameters of the dynamical map $\Lambda(t)$ are given by simple oscillations. Unlike in the previous example, the functions describing $f_{\min}$, $f_{\max}$, and $\nu_\infty$ are no longer smooth but piecewise analytic. Additionally, the conditions in the expressions for extremal fidelities depend not only on $p$ but also on the sign of the cosine function.

We plot our results in Fig.\ref{osc}. As expected, the channel measures demonstrate a similar oscillatory behaviour to that of $\lambda_1(t)$, $\lambda_3(t)$, and $\lambda_\ast(t)$. However, whereas these parameters and the extremal fidelities are $2\pi$-periodic, the maximal output norms are $\pi$-periodic instead. All the plotted curves cross at $t=k\pi$, $k\in\mathbb{N}$, with some being colinear for wider ranges of time. The discontinuity points for the extremal fidelities and $\nu_\infty$ at $p=0$ correspond to $\pi/2+k\pi$, $\pi\in\mathbb{N}$. Again, the higher the value of $p$, the smaller $f_{\min}$ and the greater the functions $f_{\max}$, $\nu_2$, $\nu_\infty$ at any fixed time. Just like for the exponentially decaying $\lambda_k(t)$, $f_{\max}$, $\nu_2$, and $\nu_\infty$ reach their maximal value of 1 for $p=1$. The only function that ever reaches zero is $f_{\min}[\Lambda_{\rm U}(t\to\infty)]$ for $\pi+2k\pi$ regardless of the choice of $p$, and then also for $\pi/2+2k\pi$ and $3\pi/2+2k\pi$ if $p=0$.

In \cite{Anindita}, non-monotonicity of the Gaussian channel fidelity implies non-Markovianity of the evolution. We obtain a similar correspondence for extremal channel fidelities, as Example 1 features a Markovian evolution (monotonic functions) and Example 2 deals with a non-Markovian evolution (non-monotonic functions) \cite{CC_GAD}.

\begin{figure}[htb!]
   \includegraphics[width=0.8\textwidth]{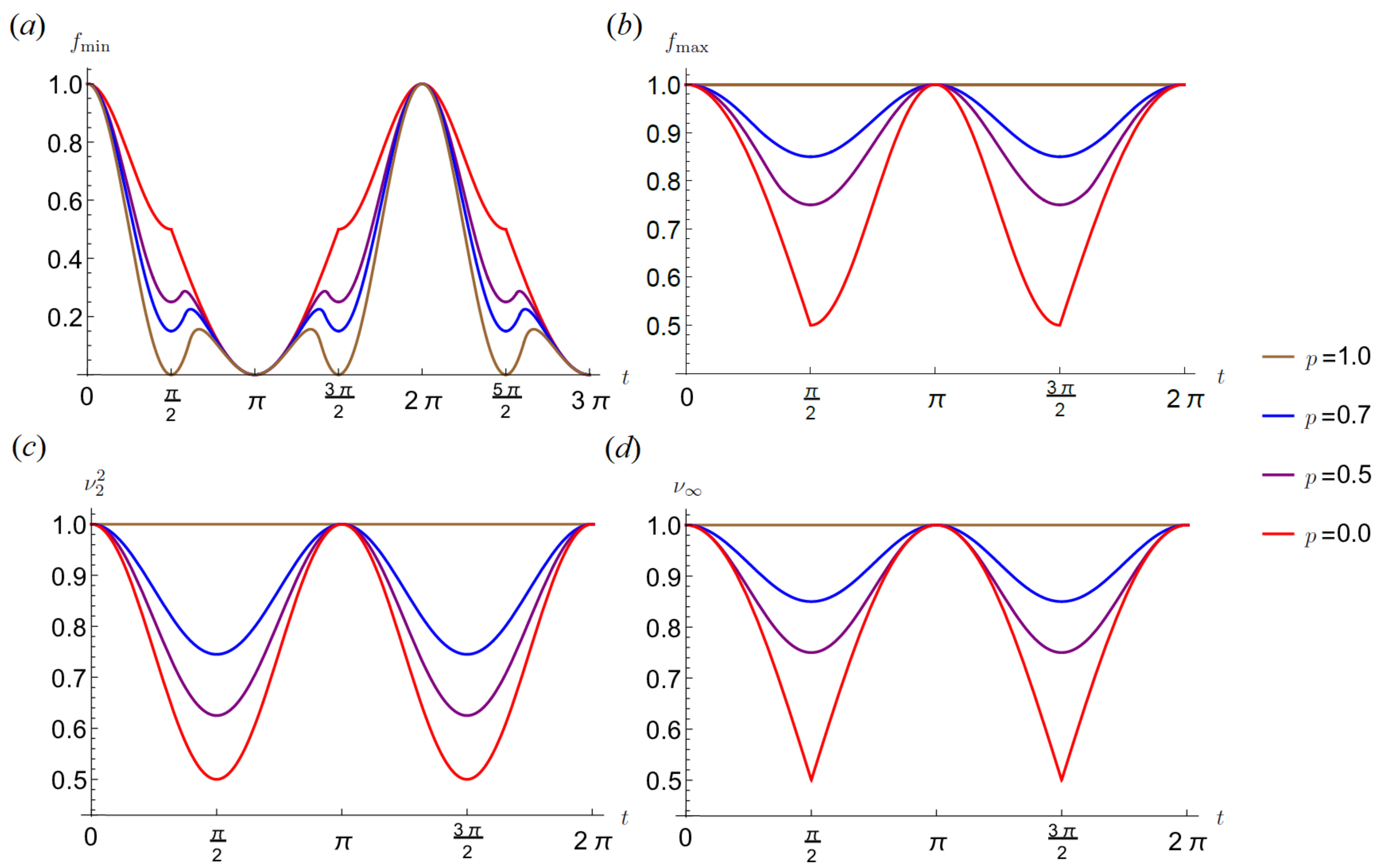}
\caption{Plots for oscillating channel eigenvalues representing time-evolution of the minimal fidelity ($a$), maximal fidelity ($b$), maximal output 2-norm ($c$), and maximal output $\infty$-norm ($d$). The curves correspond to the channel mixtures with $p=0$ (red), $p=0.3$ (black), $p=0.5$ (purple), $p=0.7$ (blue), and $p=1$ (yellow).}
\label{osc}
\end{figure}

\FloatBarrier

\subsection{Example 3 -- Entanglement death and rebirth}

\FloatBarrier

Finally, we take a closer look at what happens after sending a half of a maximally entangled qubit pair through the mixtures analyzed in the earlier examples. From eq. (\ref{crhoW}), we find that the concurrence is given by
\begin{equation}\label{c1}
c[\Lambda(t)[\rho_W]]=\max\{0,e^{-t}(1-\sqrt{1-p^2}\sinh t)\}
\end{equation}
for $\Lambda(t)$ defined via exponential functions in eq. (\ref{eigexp}) and by
\begin{equation}\label{c2}
c[\Lambda(t)[\rho_W]]=\frac 12 \max\{0,2|\cos t|-\sqrt{1-p^2}\sin^2 t\}
\end{equation}
if one instead takes the oscillating functions from eq. (\ref{eigosc}). In Fig.\ref{conc}.($a$)--($b$), we plot both functions for different values of $p$. Observe that in Fig.\ref{conc}.($a$), corresponding to exponentially decaying eigenvalues, the concurrence monotonically decays. With the increase of $p$, the entanglement of $\Lambda(t)[\rho_W]$ gets prolonged and its sudden death is postponed -- up until $p=1$, for which it eternally prevails. In Fig.\ref{conc}.($b$), on the other hand, we see that the state experiences periodical entanglement death and rebirth. The less non-unital the mixture is, the bigger the gap between entanglement death and rebirth. Moreover, the maximally entangled state is reached for $t=k\pi$, $k\in\mathbb{N}$, when $c[\Lambda(t)[\rho_W]]=1$.

\begin{figure}[htb!]
   \includegraphics[width=0.8\textwidth]{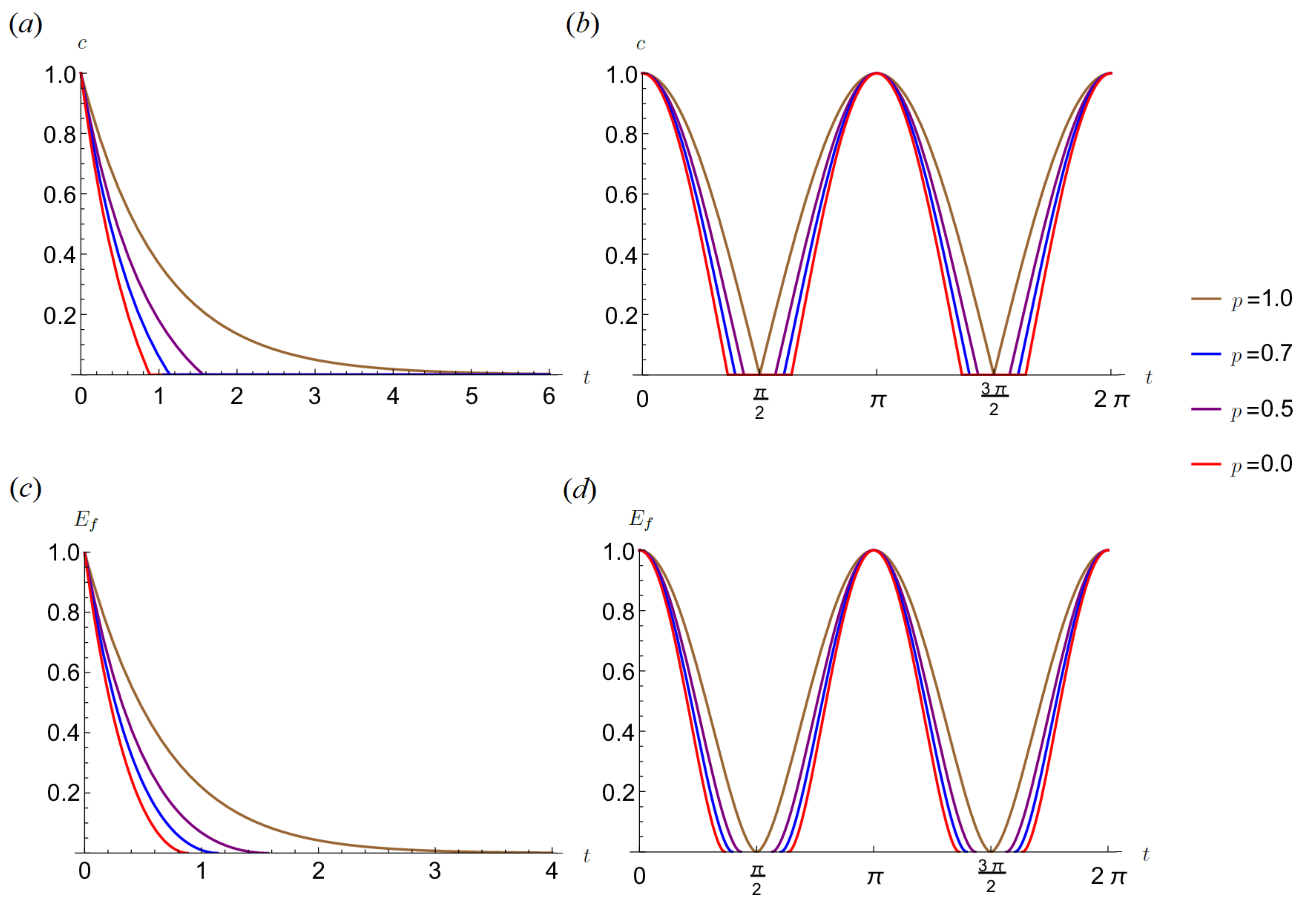}
\caption{Graphical representation of the evolution of concurrence $c[\Lambda(t)[\rho_W]]$ and the corresponding entanglement of formation $E_f(c)$ for exponentially decaying ($a$--$b$) and oscillating ($c$--$d$) channel eigenvalues. The color curves correspond to the channel mixtures with $p=0$ (red), $p=0.5$ (purple), $p=0.7$ (blue), and $p=1$ (yellow).}
\label{conc}
\end{figure}

The concurrence is directly related to another entanglement measure: entanglement of formation \cite{Wooters1,Audenaert}
\begin{equation}
E_f(\rho)=h\left[\frac{1+\sqrt{1-c^2(\rho)}}{2}\right],
\end{equation}
where
\begin{equation}
h(x)=-x\log_2x-(1-x)\log_2(1-x).
\end{equation}
The main difference between them is that only the entanglement of formation is a resource-based, information theoretic measure \cite{Wootters4}. In Fig.\ref{conc}.($c$)--($d$), we plot $E_f[\Lambda(t)[\rho_W]]$ based on the concurrence from eqs. (\ref{c1}) and (\ref{c2}), respectively. Observe that both measures reach their minimal and maximal values at the same points in time, even though their in-between values differ. Therefore, $E_f=0$ and $E_f=1$ again correspond to separable and maximally entangled states, respectively.

\FloatBarrier

\section{Conclusions}

We analyzed properties of phase-covariant channels with varying degrees of non-unitality. By fixing the channel eigenvalues and only changing its invariant state, we showed how to engineer the channel extremal fidelities, maximal output purity, and concurrence when the channel acted on one half of a maximally entangled state. We presented examples for mixing two semigroups and two non-invertible dynamical maps. Our results confirmed that, from among the measures we considered, only the minimal channel fidelity cannot be improved by introducing more non-unitality to the quantum channel. In other words, the more non-unital channels we took, the more pure and less distorted were the output states. This held true for any point in time, therefore this increase in channel performance was not only temporary, which was the case with engineering fidelity and classical capacity for unital maps \cite{Marshall,fidelity,Engineering_capacity}. Similarly, non-unital channels were better suited for prolonging quantum entanglement, even leading to its repetitive rebirth.

We claim that for non-unitality of quantum channels it is possible to formulate a research theory, similarly to non-invertibility in ref. \cite{invertibility_measure}. Recall that quantum resource theories are used to quantify desirable quantum effects, like quantum entanglement or non-Markovianity. We consider non-unitality as a dynamical quantum resource \cite{QRT} for phase-covariant dynamical maps. Unital maps can be identified with free operations, which are resource non-increasing. Our non-unitality measure from Remark 3 is a good candidate for a resource quantifier, as it is a continuous function on quantum maps that measures resourcefulness.

In further studies, it would be interesting to analyze other quantities that characterize channel performance, like von Neumann entropy or channel capacity. One could also check whether it is possible to engineer temporary increasing fidelity and purity via master equations with memory kernels, like for the Pauli channels in refs. \cite{Marshall,fidelity}. Another open question considers possible generalizations to qudit systems.

\section{Acknowledgements}

This research was funded in whole or in part by the National Science Centre, Poland, Grant numbers 2021/43/D/ST2/00102 (KS) and 2020/39/D/ST2/01234 (MS). For the purpose of Open Access, the author has applied a CC-BY public copyright licence to any Author Accepted Manuscript (AAM) version arising from this submission.

\bibliography{bibliography}
\bibliographystyle{beztytulow2}

\end{document}